\documentclass[11pt]{article}

\usepackage[margin=1in]{geometry}
\usepackage{iftex}
\usepackage{xifthen}
\usepackage{amsmath,amssymb,amsthm,mathtools}
\usepackage[T1]{fontenc}
\usepackage{lmodern}
\usepackage{microtype}
\usepackage[dvipsnames]{xcolor}
\usepackage{enumitem}
\usepackage{tcolorbox}
\tcbuselibrary{breakable}
\usepackage[numbers]{natbib}
\usepackage[hidelinks]{hyperref}
\hypersetup{
	colorlinks=true,
	pdfpagemode=UseNone,
    citecolor=OliveGreen,
    linkcolor=NavyBlue,
    urlcolor=Magenta,
	pdfstartview=FitW
}
\usepackage{prettyref}

\newrefformat{thm}{\textcolor{NavyBlue}{Theorem~\ref{#1}}}
\newrefformat{lem}{\textcolor{NavyBlue}{Lemma~\ref{#1}}}
\newrefformat{ppt}{\textcolor{NavyBlue}{Property~\ref{#1}}}
\newrefformat{pps}{\textcolor{NavyBlue}{Proposition~\ref{#1}}}
\newrefformat{def}{\textcolor{NavyBlue}{Definition~\ref{#1}}}
\newrefformat{cor}{\textcolor{NavyBlue}{Corollary~\ref{#1}}}
\newrefformat{fig}{\textcolor{NavyBlue}{Figure~\ref{#1}}}
\newrefformat{fact}{\textcolor{NavyBlue}{Fact~\ref{#1}}}
\newrefformat{claim}{\textcolor{NavyBlue}{Claim~\ref{#1}}}
\newrefformat{eq}{\textcolor{NavyBlue}{(\ref{#1})}}
\newrefformat{alg}{\textcolor{NavyBlue}{Algorithm~\ref{#1}}}
\newrefformat{cond}{\textcolor{NavyBlue}{Condition~\ref{#1}}}
\newrefformat{obs}{\textcolor{NavyBlue}{Observation~\ref{#1}}}
\newrefformat{item}{\textcolor{NavyBlue}{Item~\ref{#1}}}
\newrefformat{sec}{\textcolor{NavyBlue}{Section~\ref{#1}}}
\newrefformat{cons}{\textcolor{NavyBlue}{Constraint~\ref{#1}}}
\newrefformat{app}{\textcolor{NavyBlue}{Appendix~\ref{#1}}}
\newrefformat{item}{\textcolor{NavyBlue}{Item~\ref{#1}}}

\usepackage{amsthm, thmtools}
\usepackage{hyperref}

\usepackage[capitalise,nameinlink]{cleveref}
\crefformat{equation}{#2(#1)#3}
\Crefformat{equation}{#2(#1)#3}
\crefrangeformat{equation}{#3(#1)#4--#5(#2)#6}
\Crefrangeformat{equation}{#3(#1)#4--#5(#2)#6}

\newcommand{\addsharedenv}[1]{\declaretheorem[sibling=theorem]{#1}}

\theoremstyle{plain} 
\declaretheorem{theorem} % Base counter
\forcsvlist{\addsharedenv}{observation, claim, fact, lemma, proposition, corollary}

\theoremstyle{definition} 
\forcsvlist{\addsharedenv}{definition, remark, condition, assumption, example}
\makeatletter

\newcommand{\Rmnum}[1]{\expandafter\@slowromancap\romannumeral #1@}
\makeatother

\def\*#1{\mathbf{#1}}      % \*A for \mathbf{A}
\def\+#1{\mathcal{#1}}     % \+A for \mathcal{A}
\def\-#1{\mathrm{#1}}      % \-A for \mathrm{A}
\def\^#1{\mathbb{#1}}      % \^A for \mathbb{A}
\def\!#1{\mathfrak{#1}}    % \!A for \mathfrak{A}
\def\$#1{\mathscr{#1}}     % \$A for \mathscr{A}

\newcommand{\tuple}[2][]{ \ifthenelse{\isempty{#1}}
  {{\left(#2\right)}}
  {{\left(#2\right)}_{{#1}}} }

\newcommand{\set}[2][]{ \ifthenelse{\isempty{#1}}
  {{\left\{#2\right\}}}
  {{\left\{#2\right\}}_{{#1}}} }

\newcommand{\abs}[1]{\left\vert#1\right\vert}

\DeclareMathOperator{\oPr}{\mathbf{Pr}}
\renewcommand{\Pr}[2][]{ \ifthenelse{\isempty{#1}}
  {\oPr\left[#2\right]}
  {\oPr_{#1}\left[#2\right]} }

\DeclareMathOperator{\oE}{\mathbf{E}}
\newcommand{\E}[2][]{ \ifthenelse{\isempty{#1}}
  {\oE\left[#2\right]}
  {\oE_{#1}\left[#2\right]} }

\DeclareMathOperator{\oVar}{\mathbf{Var}}
\newcommand{\Var}[2][]{ \ifthenelse{\isempty{#1}}
  {\oVar\left[#2\right]}
  {\oVar_{#1}\left[#2\right]} }

\DeclareMathOperator{\oEnt}{\mathbf{Ent}}
\newcommand{\Ent}[2][]{ \ifthenelse{\isempty{#1}}
  {\oEnt\left[#2\right]}
  {\oEnt_{#1}\left[#2\right]} }

\newcommand{\MC}{\mathcal{MC}}

\newcommand{\per}[1]{\operatorname{per}\left(#1\right)}
\newcommand{\Cov}{\operatorname{Cov}}
\newcommand{\one}{\mathbf{1}}
\newcommand{\R}{\mathbb{R}}

\newcommand{\trel}{\tau_{\mathrm{rel}}}
\newcommand{\trelax}{\trel}

\newcommand{\tholes}{\tau_{\mathrm{holes}}}

\newcommand{\fpras}{\ensuremath{\mathsf{FPRAS}}}
\newcommand{\eps}{\varepsilon}
\newcommand{\inner}[2]{\left\langle #1,#2\right\rangle_{\pi}}

\newcommand{\e}{\mathrm{e}}
\renewcommand{\epsilon}{\varepsilon}
\renewcommand{\emptyset}{\varnothing}

\newcommand{\tp}{\tuple}

\title{Faster FPRAS for the Permanent\\ via Restricted Poincar\'e Inequalities and Coupled Flows}
\author{Xiaoyu Chen\thanks{Email: \texttt{xiaoyu@mit.edu}, 
Massachusetts Institute of Technology. Supported by the NSF CAREER grant CCF-2443045, and the Reed Fund at MIT.}%
\and Eric Vigoda\thanks{
  Emails: \texttt{\{vigoda, xiongxinyang\}@ucsb.edu},
  University of California, Santa Barbara.
  Supported in part by NSF grant CCF-2147094.
}\and Xiongxin Yang\footnotemark[2]}
\date{\today}

\begin{document}
\maketitle
\begin{abstract}
The permanent of an $n\times n$ $0/1$ matrix $A$ equals the number of perfect matchings in the bipartite graph with edges defined by $A$.  Jerrum, Sinclair, and Vigoda (2004) presented an FPRAS for approximating the permanent of any nonnegative matrix using a novel simulated-annealing algorithm.  The running time was improved by Bez\'akov\'a, \v{S}tefankovi\v{c}, Vazirani, and Vigoda (2008) to $O(n^7\log^4 n)$ for $0/1$ matrices, for any fixed approximation and success parameters.  We present the first asymptotic improvement over this running time bound, obtaining an $O(n^6\log^5 n)$-time algorithm.  As in the previous works, our algorithm extends to arbitrary nonnegative matrices.

The analysis of Bez\'akov\'a et al.\ yields an $O(n^4)$ relaxation time bound for the JSV Markov chain on perfect and near-perfect matchings with ideal hole weights, under which each hole pattern (the unmatched vertices, if any) is equally likely in the stationary distribution.  We introduce a restricted Poincar\'e inequality for the partition into hole patterns and prove an $O(n^3)$ bound on the corresponding restricted relaxation time.  Our proof uses a coupled multicommodity flow argument inspired by a recent transport-flow argument of Chen et al.~(2025) for the
Jerrum-Sinclair chain on all matchings.
\end{abstract}

\section{Introduction}
The \emph{permanent} of an $n\times n$ matrix $A=(a_{ij})$ is
\[
\per{A}=\sum_{\sigma\in S_n}\prod_{i=1}^{n}a_{i,\sigma(i)},
\]
where the sum ranges over all permutations of $[n]=\{1,\dots,n\}$.  For a $0/1$ matrix $A$, $\per{A}$ counts the perfect matchings of the bipartite graph whose bipartite adjacency matrix is $A$, equivalently, the close-packed configurations of the classical dimer model of statistical physics.  Beyond this central role in combinatorics and statistical physics, the permanent has found applications in areas such as computer vision and statistics~\cite{atanasov2016localization,oh2009markov}.
Valiant~\cite{Valiant79} proved that exact evaluation is $\#\mathsf{P}$-complete even for $0/1$ matrices.
Approximation is therefore the natural recourse.

The dominant approach has been the Markov chain Monte Carlo (MCMC) method, made possible by the polynomial-time equivalence between approximate sampling and approximate counting for self-reducible problems~\cite{JVV86}.  Broder~\cite{Broder86} proposed a natural Markov chain on all perfect and near-perfect matchings whose stationary distribution is uniform; a near-perfect matching is a matching with exactly two unmatched vertices.  Jerrum and Sinclair~\cite{JS89} analyzed the convergence of this chain and obtained an \fpras{} for every class of $0/1$ matrices in which near-perfect matchings do not outnumber perfect matchings by more than a polynomial factor; this includes all dense matrices and almost all random matrices.  Karmarkar, Karp, Lipton, Lov\'asz, and Luby~\cite{KKLLL93} gave a Monte Carlo algorithm based on an easily computed unbiased estimator, which can require exponentially many trials on adversarial inputs but performs well on almost every random matrix.  Two further lines of attack were pursued in the 1990s: Jerrum and Vazirani~\cite{JV96} obtained a subexponential algorithm running in time $\exp\bigl(O(\sqrt{n})\bigr)$, and Barvinok~\cite{Barvinok99} gave polynomial-time algorithms approximating the permanent within a simply exponential factor.  Deterministic quasipolynomial algorithms are also known~\cite{Barvinok17}, but only for matrices whose entries lie in a fixed interval $[\delta,1]$.
%: Barvinok approximates $\ln\per{A}$ within additive error $\eps$ in time $n^{O(\log n-\log\eps)}$ on such instances.

The breakthrough came with the celebrated \fpras{} of Jerrum, Sinclair, and Vigoda~\cite{JSV04}, the first to handle \emph{arbitrary} matrices with nonnegative entries.  Their algorithm is a simulated-annealing process. 
For a bipartite graph $G=(V_1\cup V_2,E)$, let $\mathcal{P}$ denote the collection of perfect matchings in the complete bipartite graph on $V_1\cup V_2$, and, for $u\in V_1$ and $v\in V_2$, let $\mathcal{N}(u,v)$ denote the collection of near-perfect matchings whose unmatched vertices are $u$ and $v$.  In graph-theoretic terms, $\per{A}$ is the number of matchings in $\mathcal{P}$ whose edges all lie in $E$.  For activities $\lambda:V_1\times V_2\rightarrow(0,1]$, a matching $M$, and a collection $S$ of matchings, write
\begin{equation*}
  \lambda(M):=\prod_{\{u,v\}\in M}\lambda(u,v),
  \qquad
  \lambda(S):=\sum_{M\in S}\lambda(M).
\end{equation*}

The simulated annealing algorithm of \cite{JSV04} starts from the complete bipartite graph by setting $\lambda(u,v)=1$ for all $u\in V_1$ and $v\in V_2$.  At this initial stage, $\lambda(\mathcal{P})=n!$.  The algorithm gradually decreases the activities of the nonedges to $1/n!$; throughout, it maintains \emph{hole weights} $w(u,v)$ approximating the \emph{ideal weights}, defined by
\begin{restatable}{equation}{idealweights}
  w^*(u,v)
    =
    \frac{\lambda(\mathcal{P})}{\lambda(\mathcal{N}(u,v))}.
\end{restatable}
With these weights, the perfect block $\mathcal{P}$ and each of the $n^{2}$ near-perfect blocks $\mathcal{N}(u,v)$ carry equal stationary mass.  Jerrum, Sinclair, and Vigoda \cite{JSV04} proved that the Metropolis chain on perfect and near-perfect matchings weighted by the ideal weights (or a constant-factor approximation) has polynomial mixing time; this yielded a polynomial-time \fpras{} for the permanent.

Building on this framework, Bez\'akov\'a, \v{S}tefankovi\v{c}, Vazirani, and Vigoda~\cite{BSVV08} redesigned the cooling schedule so that it accelerates as the temperature decreases, thereby reducing the number of annealing phases from $O(n^{2}\log n)$ to $O(n\log^{2}n)$; they also proved new structural inequalities about perfect matchings in bipartite graphs which reduced the relaxation-time bound to $O(n^4)$.  Their \fpras{} for estimating the permanent runs in time
\[
O\bigl(n^{7}\log^{4}n+n^{6}\log^{5}n\,\eps^{-2}\bigr),
\]
which remained the best known bound for nearly two decades.  
The breakdown of the running time is roughly the following: $O(n\log^2 n)$ intermediate temperatures to change the activity $\lambda(u,v)$ for each nonedge $\{u,v\}\notin E$ from $1$ to $1/n!$; $O(n^2\log n)$ samples at each temperature to recalibrate the weights; and $O(n^4\log n)$ to generate each sample after accounting for the logarithmic total-variation overhead beyond the $O(n^4)$ relaxation-time bound.

The weight-recalibration term is the bottleneck in this bound; reducing it by a factor of $n$ yields our main theorem.

\begin{theorem}
\label{thm:main-permanent}
For all $0<\varepsilon,\delta<1$, there exists a randomized algorithm that,
with probability at least $1-\delta$, approximates the permanent of an
$n\times n$ $0/1$ matrix $A$ within a factor $(1\pm\varepsilon)$ in time
$O\left({n^6}{\varepsilon^{-2}}\log^5 n\,\log(2/\delta)
\right)$.
The algorithm extends to arbitrary matrices with nonnegative entries.
\end{theorem}

To explain the source of the improvement in running time, consider the weight-recalibration step in~\cite{BSVV08}.  Given weights $w(u,v)$ on all hole patterns $(u,v)\in V_1\times V_2$ that approximate the ideal weights $w^*(u,v)$ within a constant factor, we recalibrate the weights by generating sufficiently many samples from the Markov chain and then adjusting the weights so that the stationary masses of the $n^2+1$ hole patterns are approximately uniform.  This recalibration step only observes the \emph{hole pattern} of each sampled matching, namely, whether it is perfect, and if not, which pair $(u,v)\in V_1\times V_2$ is unmatched; it does not consider the entire matching.

Motivated by the local Poincar\'e inequalities of Chen et
al.~\cite{CFJ25}, we consider a Poincar\'e inequality restricted to a partition of the state space; for the application to the permanent, the partition is the $n^2+1$ hole patterns.  Whereas the relaxation time $\trelax$ (or equivalently the spectral gap) measures the decay rate of variance for any functional on the entire state space, the new restricted Poincar\'e inequality only considers functions that are averaged within each block of the partition (e.g., hole pattern).  We denote this weaker notion of relaxation time, which is \emph{restricted} to hole-pattern observables, by $\tholes$.  We prove that $\tholes=O(n^{3})$, improving on the best known bound of $\trelax=O(n^{4})$ for the full chain.  As a consequence, we obtain a factor-$n$ improvement in the correlation scale governing the hole-pattern frequency estimates in the \cite{BSVV08} weight-refinement step; this gives \cref{thm:main-permanent}.

It remains to bound this restricted relaxation time.  For this purpose, we refine the classical multicommodity-flow
argument used to bound the spectral gap of the chain.  A classical canonical-path or multicommodity-flow analysis of this chain must consider all pairs of states independently.  For instance, for a near-perfect matching $I$ and a perfect matching $F$, the flow from $I$ to $F$ follows the augmenting path and the even alternating cycles in their symmetric difference.  Here, we couple the conditional distribution of each near-perfect block $\mathcal{N}(u,v)$ with that of the perfect block $\mathcal{P}$ and route a canonical path between the two coupled matchings only along their unique augmenting path.  By avoiding the even cycles entirely, we save a factor of order $n$ in the subsequent congestion bound.  Our construction is inspired by the coupling-based canonical-path arguments developed in~\cite{CCFV2025,CFJ25}, though the setting and the functional inequality we target are different.

\paragraph{Organization.}
In the following section, we present an overview of the \fpras{} for the permanent and our improved running time.  In \cref{sec:restricted-poincare}, we recall the classical Poincar\'e inequality and relaxation time, introduce their restricted analogues, and present the coupled multicommodity-flow technique for bounding the restricted relaxation time.  In \cref{sec:improved-random-weight-refinement,sec:bound-for-congestion}, we apply this technique to the permanent algorithm.

\paragraph{AI Assistance.}
The high-level algorithmic ideas originated with the authors. GPT-5.6 Pro was used in exploratory discussions that helped formulate the partition-restricted relaxation time and conceptualize the general restricted-Poincaré and coupled-flow framework presented in \cref{sec:restricted-poincare}. The core text, final proofs, and permanent-specific arguments were written independently by the human authors, with ChatGPT utilized during the writing process to refine phrasing and provide editorial suggestions. The authors independently verified all mathematical claims and take full responsibility for the final content.

\section{Sketch of the Permanent Algorithm}
\label{sec:overview}

We now give a sketch of the permanent algorithm, which proves \cref{thm:main-permanent}. %}
We first recall the annealing framework of \cite{JSV04,BSVV08} and then isolate the weight-refinement step where our speedup occurs.
For simplicity, consider the permanent of an $n\times n$ $0/1$ matrix $A$, and let $G=(V_1\cup V_2,E)$ be the corresponding (unweighted) bipartite graph where $|V_1|=|V_2|=n$.

Let $\mathcal P$ be the set of perfect matchings in the complete bipartite graph on $(V_1,V_2)$.  For $u\in V_1$ and $v\in V_2$, let $\mathcal N(u,v)$ be the set of near-perfect matchings whose unmatched vertices, or \emph{holes}, are $u$ and $v$, and set $\Omega:=\mathcal P\cup\bigcup_{(u,v)\in V_1\times V_2}\mathcal N(u,v)$.
For a scalar activity $\lambda\in[0,1]$, assign activity $1$ to every edge of $G$ and activity $\lambda$ to every nonedge.  Extend these activities multiplicatively to a matching $M$ and additively to a set $S\subseteq\Omega$ by
\begin{equation*}
  \lambda(M):=\prod_{e\in M}\lambda(e),
  \qquad
  \lambda(S):=\sum_{M\in S}\lambda(M).
\end{equation*}
At the initial activity $\lambda_0=1$, the underlying weighted graph is complete, so $\lambda_0(\mathcal P)=n!$.  
The annealing process then decreases the activity through a sequence of positive values $\lambda_0,\lambda_1,\ldots,\lambda_\ell$, called the \emph{cooling schedule}, where $\lambda_\ell$ is sufficiently close to zero that $\lambda_\ell(\mathcal P)$ approximates $\operatorname{per}(A)$.
By the telescoping identity,
\begin{equation}
  \operatorname{per}(A)
  \approx
  \lambda_\ell(\mathcal P)
  =
  \lambda_0(\mathcal P)
  \prod_{i=0}^{\ell-1}
  \frac{\lambda_{i+1}(\mathcal P)}{\lambda_i(\mathcal P)}
  =
  n!
  \prod_{i=0}^{\ell-1}
  \frac{\lambda_{i+1}(\mathcal P)}{\lambda_i(\mathcal P)}.
  \label{eq:annealing-telescoping}
\end{equation}
Thus it suffices to estimate the ratio between the weighted numbers of perfect matchings at consecutive activities.

These ratios can be estimated using the JSV chain of~\cite{JSV04}, an ergodic Metropolis chain move between perfect and near-perfect matchings that will be defined formally in \cref{sec:metropolis-chain}.  
At phase $i$, assign a positive hole weight $w_i(u,v)$ to each near-perfect block and give a matching $M$ weight $\lambda_i(M)$ if $M\in\mathcal P$, and weight $w_i(u,v)\lambda_i(M)$ if $M\in\mathcal N(u,v)$.  If $\pi_{\lambda_i,w_i}$ denotes the resulting stationary distribution, then its conditional distribution on the perfect-matching block is
$\pi_{\lambda_i,w_i}(M\mid\mathcal P)=\frac{\lambda_i(M)}{\lambda_i(\mathcal P)}$ for $M\in\mathcal P$, independently of the choice of hole weights.  
Consequently,
\begin{equation*}
  \frac{\lambda_{i+1}(\mathcal P)}{\lambda_i(\mathcal P)}
  =
  \E[M\sim\pi_{\lambda_i,w_i}]{
    \left.
    \frac{\lambda_{i+1}(M)}{\lambda_i(M)}
    \,\right|\,
    M\in\mathcal P
  }.
\end{equation*}
Hence estimating the product in \cref{eq:annealing-telescoping} requires the chain to produce sufficiently many perfect-matching samples throughout the cooling schedule.

Although the conditional distribution above does not depend on the hole weights, the mixing of the full chain and the frequency with which it visits $\mathcal P$ do.  
The hole weights are therefore chosen to balance the stationary masses of the perfect and near-perfect blocks.  
For a fixed activity $\lambda$, the ideal hole weights are
\idealweights*

\noindent Under these weights, all $n^2+1$ hole-pattern blocks have the same stationary probability:
\begin{equation}
  \pi_{\lambda,w^*}(\mathcal N(u,v))
  =
  \pi_{\lambda,w^*}(\mathcal P)
  =
  \frac1{n^2+1}.
  \label{eq:ideal-weights-uniform}
\end{equation}
Moreover, the polynomial mixing guarantee of~\cite{JSV04} remains valid when the hole weights approximate the ideal weights within a constant factor.

At the initial activity, the ideal weights are known explicitly: $w_0^*(u,v)=n$ because $\lambda_0(\mathcal P)=n!$ and $\lambda_0(\mathcal N(u,v))=(n-1)!$.  As the activity changes, however, the ideal weights change as well.  The cooling schedule is chosen so that weights that are accurate at one activity remain rough at the next.  At each new activity, the algorithm refines these inherited rough weights into accurate weights before proceeding.  Weight refinement thereby maintains, throughout the schedule, the efficient sampling guarantee required by the telescoping estimator.

It remains to explain how to perform this refinement without knowing $\lambda(\mathcal P)$ or $\lambda(\mathcal N(u,v))$.  For the chain at a fixed activity $\lambda$ with current hole weights $w$,
\begin{equation*}
  \pi_{\lambda,w}(\mathcal P)
  \propto
  \lambda(\mathcal P),
  \qquad
  \pi_{\lambda,w}(\mathcal N(u,v))
  \propto
  w(u,v)\lambda(\mathcal N(u,v)),
\end{equation*}
where the same proportionality factor applies to every block.  Taking the ratio gives
\begin{equation}
  w^*(u,v)
  =
  \frac{\lambda(\mathcal P)}{\lambda(\mathcal N(u,v))}
  =
  w(u,v)
  \frac{\pi_{\lambda,w}(\mathcal P)}
       {\pi_{\lambda,w}(\mathcal N(u,v))}.
  \label{eq:refining-weights}
\end{equation}
Thus weight refinement reduces to simultaneously estimating the $n^2+1$ stationary block probabilities.  Relative approximations to these probabilities yield relative approximations to every ratio in \cref{eq:refining-weights}, and hence to every ideal hole weight.  We now formalize the approximation guarantees required of this refinement step.

A collection of hole weights $w$ is called \emph{rough} if
$w(u,v)/w^*(u,v)\in [1/2, 2]$ for every $(u,v)\in V_1\times V_2$,
and \emph{accurate} if $w(u,v)/w^*(u,v)\in [1/\sqrt{2}, \sqrt{2}]$ for every $(u,v)\in V_1\times V_2$. 
A randomized algorithm $\mathcal R$ is called a \emph{random weight-refinement algorithm} if, 
given the current activities $\lambda$, a collection of \emph{rough} hole weights $w$, 
and a failure parameter $0<\delta<1$, it outputs, with probability at least $1-\delta$, 
a collection of weights $w'$ that is \emph{accurate} with respect to the same activities.  
Its running time is denoted by $T_{\mathcal R}(n,\delta)$.

The speedup in our FPRAS comes from a faster random weight-refinement algorithm.  
For the JSV chain with rough hole weights, \cite{BSVV08} proved that $\trelax=O(n^4)$, thereby controlling the temporal correlations of arbitrary observables of the full matching.  
Weight refinement, however, uses only observables determined by the hole pattern,
such as in the estimator corresponding to \cref{eq:refining-weights}.
This motivates the hole-restricted relaxation time $\tholes$, formally defined in \cref{sec:restricted-poincare}.  
We prove in \cref{lem:restricted-poincare} that $\tholes=O(n^3)$, reducing the correlation scale relevant to weight refinement, and hence its running time, by a factor of $n$.
\begin{theorem}
    \label{thm:hole-weight-boosting}
    Suppose that $n\geq 2$ and that $G$ contains at least one perfect matching.  
    There exists a random weight-refinement algorithm $\mathcal R$ that applies to every choice of positive activities and, for every failure parameter $0<\delta<1$, runs in time
    \begin{equation*}
      T_{\mathcal R}(n,\delta)
      =O\left(
        n^5\log n\,\log\left(\frac{2n}{\delta}\right)
      \right).
    \end{equation*}
\end{theorem}

\begin{proof}[Proof of \cref{thm:main-permanent}]
  The cases $n=1$ and $\operatorname{per}(A)=0$ are handled directly, the latter being detected by bipartite matching.  We may therefore assume that $n\geq2$ and that $G$ contains a perfect matching.

  Let $\mathcal R$ be the weight-refinement algorithm from \cref{thm:hole-weight-boosting}, and use it within the accelerated cooling procedure of~\cite{BSVV08}.  The explicit bound in the proof of~\cite[Lemma~3.1]{BSVV08}, with $s=n$, $c=\sqrt2$, $\gamma=n!$, and $D=1$, gives at most $48n(\log n)^2$ activity values, and hence at most that many refinement rounds.  At each activity, including the final activity $1/n!$, apply $\mathcal R$ to replace the current rough weights by accurate weights and store the result.  By~\cite[Lemma~3.2 and equations~(8.1)--(8.2)]{BSVV08}, every ideal hole weight changes by at most a factor $\sqrt2$ between consecutive activities.  Hence, if $w'$ is accurate at the current activity, then at the next activity
  \begin{equation*}
    \frac{w'(u,v)}{w^*_{\mathrm{next}}(u,v)}
    =
    \frac{w'(u,v)}{w^*_{\mathrm{current}}(u,v)}
    \frac{w^*_{\mathrm{current}}(u,v)}{w^*_{\mathrm{next}}(u,v)}
    \in\left[\frac12,2\right].
  \end{equation*}
  Thus accurate weights at one activity are rough at the next, and the hypothesis of \cref{thm:hole-weight-boosting} propagates throughout the schedule.

  Set the failure probability of each refinement call to $\eta:=1/(384n(\log n)^2)$.  Conditional on all preceding calls succeeding, the next input is rough by the preceding argument.  A first-failure union bound therefore shows that all calls succeed with probability at least $7/8$.  Moreover,
  \begin{equation*}
    T_{\mathcal R}(n,\eta)
    =O(n^5\log^2 n),
  \end{equation*}
  so the total cost of weight refinement is $O(n^6\log^4 n)$.

  Conditional on successful refinement, the telescoping-product estimator of~\cite[Section~4.5]{BSVV08} returns a multiplicative $(1\pm\varepsilon)$ approximation in time $O(n^6\log^5 n\,\varepsilon^{-2})$ with probability at least $7/8$, after a constant number of repetitions.  Thus one complete execution succeeds with probability at least $3/4$ and runs in time $O(n^6\log^5 n\,\varepsilon^{-2})$.
  Taking the median of $O(\log(2/\delta))$ independent executions reduces the failure probability to at most $\delta$, giving total running time
  \begin{equation*}
    O\left(
      \frac{n^6}{\varepsilon^2}
      \log^5 n\,
      \log\left(\frac{2}{\delta}\right)
    \right).
  \end{equation*}
  The extension to nonnegative matrices follows from the weighted cooling reduction of~\cite[Section~9]{BSVV08}, since \cref{thm:hole-weight-boosting} applies to arbitrary positive activities.
\end{proof}

\section{Poincar\'e and Restricted Poincar\'e Inequalities} \label{sec:restricted-poincare}

The algorithmic improvement comes from distinguishing relaxation of the full matching chain from the relaxation visible through a prescribed partition—namely, 
the partition by hole patterns in our application to the permanent approximation. 
We begin by recalling the standard Poincar\'e inequality and then develop an abstract partition-restricted analogue. 
After deriving the corresponding sampling guarantees, we explain how to bound the standard and restricted relaxation times using multicommodity flows and coupled flows, respectively.

\subsection{The Standard Poincar\'e Inequality}

Throughout this section, 
let $P$ be the transition matrix of a finite irreducible reversible Markov chain on $\Omega$, and let $\pi$ be its stationary distribution.  
For a function $f:\Omega\to\R$, write
\begin{equation*}
  \E[\pi]{f}:=\sum_{x\in\Omega}\pi(x)f(x).
\end{equation*}
We also regard $P$ as an operator on functions, so that $(Pf)(x):=\sum_{y\in\Omega}P(x,y)f(y)$.
The global fluctuation of $f$ under stationarity is measured by the \emph{variance},
which is defined as
\begin{equation*}
  \Var[\pi]{f}
  :=\sum_{x\in\Omega}\pi(x)\bigl(f(x)-\E[\pi]{f}\bigr)^2
  =\frac12\sum_{x,y\in\Omega}\pi(x)\pi(y)\bigl(f(x)-f(y)\bigr)^2.
\end{equation*}
By contrast, the \emph{Dirichlet form}
\begin{equation*}
  \+E_P(f,f)
  :=\frac12\sum_{x,y\in\Omega}\pi(x)P(x,y)\bigl(f(x)-f(y)\bigr)^2
\end{equation*}
measures the variation seen across a single transition of the chain.
The Poincar\'e inequality asserts that this local variation controls the full stationary variance.

\begin{definition}[Poincar\'e inequality]
    We say the Markov chain $P$ with stationary distribution $\pi$ satisfies the \emph{Poincar\'e inequality} with constant $C$ if
    \begin{equation*}
      \forall f:\Omega\to\^R,\qquad
      \Var[\pi]{f}\leq C\,\+E_P(f,f).
    \end{equation*}
    The \emph{relaxation time} of $P$, denoted by $\trel$, is the minimum constant $C$ satisfying the Poincar\'e inequality.  Equivalently,
    \begin{equation*}
      \trel
      :=\sup_{\substack{f:\Omega\to\R\\\Var[\pi]{f}>0}}
      \frac{\Var[\pi]{f}}{\+E_P(f,f)}.
    \end{equation*}
\end{definition}
  
The full relaxation time is governed by the worst observable, including fluctuations that remain entirely within one block and are invisible to the statistic being estimated. 
To discard these fluctuations, we next average observables within the blocks of a fixed partition.

\subsection{The Restricted Poincar\'e Inequality}

Let $\+S=\{S_1,\ldots,S_k\}$ be a partition of $\Omega$, whose elements we call \emph{blocks}.  For $1\leq i\leq k$, let $\pi_i:=\pi(\,\cdot\mid S_i)$ be the stationary distribution conditioned on $S_i$.  For an arbitrary function $f:\Omega\to\R$, define its block average by
\begin{equation*}
  (\mathsf A_{\+S}f)(x)
  :=\E[X\sim\pi]{f(X)\mid X\in S_i}
  =\E[\pi_i]{f}
  \qquad (x\in S_i).
\end{equation*}
Thus $\mathsf A_{\+S}f$ retains the conditional mean of $f$ on each block and discards all variation within that block.  The law of total variance makes this separation explicit:
\begin{equation*}
  \Var[\pi]{f}
  =\sum_{i=1}^k \pi(S_i)\Var[\pi_i]{f}
  +\Var[\pi]{\mathsf A_{\+S}f}.
\end{equation*}
The first term records fluctuations within the blocks, 
whereas the second records only the variation among their conditional means.  We denote the latter quantity by
\begin{equation}
  \Var[\+S]{f}
  :=\Var[\pi]{\mathsf A_{\+S}f}
  =\frac12\sum_{i,j=1}^k \pi(S_i)\pi(S_j)
  \left(\E[\pi_i]{f}-\E[\pi_j]{f}\right)^2.
  \label{eq:restricted-variance-pairwise}
\end{equation}
This is the variance relevant to an observer who sees the block containing the state but not the state itself.  
It therefore suggests the following restricted analogue of the standard Poincar\'e inequality.

\begin{definition}[Restricted Poincar\'e inequality]
\label{def:restricted-relax}
  We say the Markov chain $P$ with stationary distribution $\pi$ satisfies the \emph{restricted Poincar\'e inequality} with constant $C$ under the partition $\+S$ if
  \begin{equation*}
    \forall f:\Omega\to\^R,\qquad
    \Var[\+S]{f}\leq C\,\+E_P(f,f).
  \end{equation*}
  The \emph{restricted relaxation time} under $\+S$, denoted by $\tau_{\+S}$, is the minimum constant $C$ satisfying the restricted Poincar\'e inequality.  Equivalently,
  \begin{equation*}
    \tau_{\+S} :=\sup_{\substack{f:\Omega\to\^R\\\Var[\+S]{f}>0}} \frac{\Var[\+S]{f}}{\+E_P(f,f)}.
  \end{equation*}
\end{definition}

The total-variance decomposition gives $\Var[\+S]{f}\leq\Var[\pi]{f}$ and hence $\tau_{\+S}\leq\trel$.  
Equality holds for the partition into singleton blocks, because then $\mathsf A_{\+S}$ is the identity operator.  
For a coarser partition, however, $\tau_{\+S}$ may be much smaller: it ignores slow modes that change the state without changing its block.

\subsection{Sampling Consequences}

To connect these inequalities with sampling, let
\begin{equation*}
  \+F_{\+S}
  :=\set{g:\Omega\to\R: g(x)=g(y)\text{ whenever }x,y\in S_i\text{ for some }i\in [k]}
\end{equation*}
be the space of block-constant observables.  For every $f\in\+F_{\+S}$, we have $\mathsf A_{\+S}f=f$ and hence $\Var[\+S]{f}=\Var[\pi]{f}$.
We additionally assume in this subsection that the spectrum of $P$ is nonnegative.  
In particular, this condition holds when the chain is lazy, meaning that $P(x,x)\geq1/2$ for every $x\in\Omega$.  
The standard relaxation time controls the dependence among samples of an arbitrary observable.  
The restricted relaxation time provides the analogous, potentially sharper control for block-constant observables.  
The following theorem makes this comparison precise.

\begin{theorem}[Variance of stationary empirical averages]
  \label{thm:restricted-resampling}
  Let $X_0\sim\pi$, and let $X_1,\ldots,X_{m-1}$ be generated by a finite irreducible reversible Markov chain whose spectrum is nonnegative and whose stationary distribution is $\pi$.
  \begin{enumerate}[label=(\roman*)]
    \item\label{item:restricted-resampling-standard} For every $f:\Omega\to\R$ and every integer $m\geq1$,
    \begin{equation*}
      \Var{\frac1m\sum_{t=0}^{m-1}f(X_t)}
      \leq
      \frac{2\trel}{m}\Var[\pi]{f}.
    \end{equation*}
    \item\label{item:restricted-resampling-restricted} For every partition $\+S$ of $\Omega$, every $g\in\+F_{\+S}$, and every integer $m\geq1$,
    \begin{equation*}
      \Var{\frac1m\sum_{t=0}^{m-1}g(X_t)}
      \leq
      \frac{2\tau_{\+S}}{m}\Var[\pi]{g}.
    \end{equation*}
  \end{enumerate}
\end{theorem}

The proof of \cref{thm:restricted-resampling}, deferred to \Cref{sec:proof-restricted-resampling}, bounds the integrated covariance of a block-constant observable using the restricted Poincar\'e inequality.
It is worth noting that \cref{item:restricted-resampling-restricted} generalizes \cref{item:restricted-resampling-standard}: the latter is obtained by taking $\+S$ to be the partition of $\Omega$ into singleton blocks.
We state \cref{item:restricted-resampling-standard} separately to emphasize the comparison with the standard relaxation time.
  
The conclusion of \cref{item:restricted-resampling-restricted} does not assert that the full chain mixes within $\tau_{\+S}$ steps.  
Rather, it bounds the cumulative temporal dependence visible to block-constant observables: a trajectory of length $m$ provides an effective sample size of order $m/\tau_{\+S}$ for such observables.  
This is precisely the weaker form of decorrelation needed when the estimator records only the block visited by the chain.  This is captured more concretely in the following corollary.

\begin{corollary}[Mean estimation with stationary initialization]
    \label{cor:restricted-mean-estimation}
    Consider a finite irreducible reversible Markov chain with nonnegative spectrum and stationary distribution $\pi$.
    Let $\+S=\{S_1,\ldots,S_k\}$ be a partition of $\Omega$.
    Let $0<\varepsilon,\delta\leq1$.
    \begin{enumerate}[label=(\alph*)]
      \item Let $f:\Omega\to[0,1]$ satisfy $\E[\pi]{f}>0$.  From independent trajectories, each initialized with $X_0\sim\pi$, one can obtain a relative-$\varepsilon$ estimate of $\E[\pi]{f}$ with failure probability at most $\delta$ and total trajectory length
      \begin{equation*}
        m=O\tuple{
          \frac{\trel}
          {\E[\pi]{f}\varepsilon^2}
          \log(2/\delta)}.
      \end{equation*}
  
      \item\label{item:restricted-mean-estimation} 
      Let $g\in\+F_{\+S}$ take values in $[0,1]$ and satisfy $\E[\pi]{g}>0$.  
      From independent trajectories, each initialized with $X_0\sim\pi$, one can obtain a relative-$\varepsilon$ estimate of $\E[\pi]{g}$ with failure probability at most $\delta$ and total trajectory length
      \begin{equation*}
       m=O\tuple{
          \frac{\tau_{\+S}}
          {\E[\pi]{g}\varepsilon^2}
          \log(2/\delta)}.
      \end{equation*}

    \end{enumerate}
\end{corollary}

\begin{proof}
  We first prove \cref{item:restricted-mean-estimation}.
  Write $\mu=\E[\pi]{g}$ and consider the empirical mean along a trajectory of length $m_0$ starting from $X_0\sim\pi$.
  By \cref{item:restricted-resampling-restricted} of \cref{thm:restricted-resampling} and the bound $\Var[\pi]{g}\leq\mu$,
  \begin{equation*}
    \Var{\frac1{m_0}\sum_{t=0}^{m_0-1}g(X_t)}
    \leq
    \frac{2\tau_{\+S}\mu}{m_0}.
  \end{equation*}
  Taking a trajectory of length $m_0\geq16\tau_{\+S}/(\mu\varepsilon^2)$ and applying Chebyshev's inequality gives relative error at most $\varepsilon$ with failure probability at most $1/8$.
  Repeating this independently $O(\log(2/\delta))$ times and taking the median gives the stated total trajectory length and proves \cref{item:restricted-mean-estimation}.
  Part~(a) follows identically, using \cref{item:restricted-resampling-standard} and $\trel$ in place of \cref{item:restricted-resampling-restricted} and $\tau_{\+S}$.
\end{proof}

In contrast to the restricted relaxation time introduced above, 
it is intuitively appealing to consider the relaxation time of the corresponding projection chain.  
However, as pointed out in the following remark, the projection time is not the correct notion to consider in this context.

\begin{remark}[Comparison with projection chain]
    Given a partition $\+S=\{S_1,\ldots,S_k\}$ of $\Omega$, one may form the projection chain that records only the block containing the current state.  
    Its transition matrix is
    \begin{equation*}
      \overline P(i,j)
      :=
      \sum_{x\in S_i}\pi_i(x)P(x,S_j),
      \qquad
      P(x,S_j):=\sum_{y\in S_j}P(x,y).
    \end{equation*}
    The projection chain is reversible with stationary distribution $(\pi(S_i))_{i=1}^k$.  
    Lifting a function on the blocks to a block-constant function on $\Omega$ gives the variational representation
    \begin{equation*}
      \tau_{\mathrm{proj}(\+S)}
      =
      \sup_{\substack{
        g\in\+F_{\+S}\\
        \Var[\pi]{g}>0
      }}
      \frac{\Var[\pi]{g}}{\+E_P(g,g)}.
    \end{equation*}
    Consequently,
    \begin{equation*}
      \tau_{\mathrm{proj}(\+S)}
      \leq
      \tau_{\+S}
      \leq
      \trel.
    \end{equation*}
    The first inequality follows by restricting the variational definition of $\tau_{\+S}$ to block-constant functions, for which $\Var[\+S]{g}=\Var[\pi]{g}$.  
    The second follows from $\Var[\+S]{f}\leq\Var[\pi]{f}$.
  
    The projection-chain relaxation time alone does not imply the sampling bound in \cref{item:restricted-resampling-restricted} of \cref{thm:restricted-resampling}.
    Although $\overline P$ gives the correct one-step transition probabilities between blocks, the observed process of block labels need not be Markov.  
    Conditional on the entire observed history, the hidden state in the current block $S_i$ need not have distribution $\pi_i$.  
    In effect, the projection chain averages away this hidden state after every coarse transition, whereas the original chain may retain information within the block.  
    The definition of $\tau_{\+S}$ captures this hidden memory by allowing the test function to depend on the full state while retaining only its block average in the variance.
\end{remark}

\subsection{Multicommodity Flows}
\label{subsec:multicommodity-flows}

The sampling consequences above reduce the algorithmic problem to bounding relaxation times. 
We begin with the classical \emph{multicommodity-flow} method for $\trel$: it routes the pairwise differences appearing in the variance through transitions of the chain and compares the induced load with the Dirichlet form. 
This formulation will then be adapted to the restricted setting.

Let
\begin{equation*}
  \+T:=\set{\{x,y\}\subseteq\Omega:x\neq y\text{ and }P(x,y)>0}
\end{equation*}
be the set of undirected non-loop transitions.  For $e=\{x,y\}\in\+T$, reversibility allows us to define its capacity by
\begin{equation*}
  Q(e):=\pi(x)P(x,y)=\pi(y)P(y,x).
\end{equation*}
Because self-loops make no contribution, the Dirichlet form can be written as
\begin{equation}
  \+E_P(f,f)
  =
  \sum_{e=\{x,y\}\in\+T}Q(e)\bigl(f(x)-f(y)\bigr)^2.
  \label{eq:flow-dirichlet}
\end{equation}

For distinct $x,y\in\Omega$, let $\operatorname{Paths}(x,y)$ be the collection of simple paths from $x$ to $y$ in the transition graph.  
We may restrict attention to simple paths,
since removing a cycle from a path can only decrease its length and its load on every transition.

A multicommodity flow $\Phi$ assigns a nonnegative value $\Phi_{x,y}(\gamma)$ to every $\gamma\in\operatorname{Paths}(x,y)$ and every unordered pair $\{x,y\}\in\binom{\Omega}{2}$.  
The flow must route the demand $\pi(x)\pi(y)$ between each pair, i.e.,
\begin{equation}
  \sum_{\gamma\in\operatorname{Paths}(x,y)}
  \Phi_{x,y}(\gamma)
  =
  \pi(x)\pi(y).
  \label{eq:standard-flow-demand}
\end{equation}
The demand may be divided among several paths.
The \emph{canonical-paths} method is the special case in which all of the demand $\pi(x)\pi(y)$ is routed along one path $\gamma_{x,y}$.

The length-weighted \emph{congestion} of $\Phi$ is
\begin{equation*}
  \rho(\Phi)
  :=
  \max_{e\in\+T}
  \frac1{Q(e)}
  \sum_{\{x,y\}\in\binom{\Omega}{2}}
  \sum_{\substack{\gamma\in\operatorname{Paths}(x,y)\\e\in\gamma}}
  \Phi_{x,y}(\gamma)\abs{\gamma}.
\end{equation*}
Sinclair's standard multicommodity-flow bound~\cite{Sinclair92} asserts that the relaxation time is at most the congestion of any feasible multicommodity flow.
\begin{lemma}[\cite{Sinclair92}]
  \label{lem:sinclair-flow}
  For every multicommodity flow satisfying \cref{eq:standard-flow-demand},
  \begin{equation*}
    \trel\leq\rho(\Phi).
  \end{equation*}
\end{lemma}

In previous applications of multicommodity flows to the permanent~\cite{JSV04,BSVV08}, 
it is important to designate a target set $A$; 
in the permanent setting, $A$ is the set of all perfect matchings.  
Instead of defining a flow between every pair $x,y\in\Omega$, we consider only pairs with $x\in\Omega$ and $y\in A$, at the cost of an extra factor of $\pi(A)^{-1}$ in the congestion bound.  
This restriction is crucial in~\cite{JSV04,BSVV08} for designing flows between pairs of near-perfect matchings to stay within $\Omega$.

\begin{corollary}[Restricted terminal set]
  \label{cor:restricted-terminal-flow}
  Let $A\subseteq\Omega$ be nonempty.  Suppose that, for every $x\in\Omega$ and $y\in A$ with $x\neq y$, a flow $\Phi_{x,y}$ routes demand $\pi(x)\pi(y)$ from $x$ to $y$.  Define
  \begin{equation*}
    \rho_A(\Phi)
    :=
    \max_{e\in\+T}
    \frac1{Q(e)}
    \sum_{x\in\Omega}
    \sum_{\substack{y\in A\\y\neq x}}
    \sum_{\substack{\gamma\in\operatorname{Paths}(x,y)\\e\in\gamma}}
    \Phi_{x,y}(\gamma)\abs{\gamma}.
  \end{equation*}
  Then
  \begin{equation*}
    \trel
    \leq
    \frac{\rho_A(\Phi)}{\pi(A)}.
  \end{equation*}
\end{corollary}

\begin{proof}
  By the variational characterization of variance and Jensen's inequality,
  \begin{align*}
    \Var[\pi]{f}
    &\leq
    \sum_{x\in\Omega}\pi(x)
    \left(f(x)-\E[\pi(\,\cdot\mid A)]{f}\right)^2 \\
    &\leq
    \frac1{\pi(A)}
    \sum_{x\in\Omega}\sum_{y\in A}
    \pi(x)\pi(y)\bigl(f(x)-f(y)\bigr)^2.
  \end{align*}
  Route each term along the prescribed flow.  Applying Cauchy--Schwarz along each path and then the definition of $\rho_A(\Phi)$ gives
  \begin{equation*}
    \Var[\pi]{f}
    \leq
    \frac{\rho_A(\Phi)}{\pi(A)}\+E_P(f,f).
  \end{equation*}
  Taking the supremum over $f$ proves the result.
\end{proof}

\subsection{Coupled Flows for a Partition-Restricted Poincar\'e Inequality}
\label{subsec:partition-coupled-flows}

An ordinary multicommodity flow prescribes the demand between each pair of states. 
By contrast, the restricted variance determines only the demand between blocks, leaving open how this demand is coupled among individual endpoints. 
This flexibility leads to the following \emph{coupled-flow} formulation, closely related to the transport-flow viewpoint of~\cite{CFJ25}. 
However, once the coupling fixes the endpoint demand matrix, our argument follows directly from Sinclair's path-congestion argument
without requiring an additional local-to-global framework as in \cite{CFJ25}.

Recall that $\+S=\{S_1,\ldots,S_k\}$ is a partition of $\Omega$, with $\pi_i=\pi(\,\cdot\mid S_i)$ for every $1\leq i\leq k$.  
For each pair of distinct blocks $S_i$ and $S_j$, the restricted variance assigns a total demand of $\pi(S_i)\pi(S_j)$ between the two blocks.  
Unlike the ordinary multicommodity flow, which fixes the demand between individual endpoints $x$ and $y$ to be $\pi(x)\pi(y)$, the restricted variance does not prescribe how this block-level demand is distributed among pairs $(x,y)\in S_i\times S_j$.  
It requires only that the marginal distribution of each endpoint be the corresponding conditional stationary distribution.  
We exploit this freedom by choosing, for every $1\leq i<j\leq k$, a coupling $\kappa_{ij}$ of $\pi_i$ and $\pi_j$; that is, a probability distribution on $S_i\times S_j$ satisfying
\begin{equation*}
  \sum_{y\in S_j}\kappa_{ij}(x,y)=\pi_i(x),
  \qquad
  \sum_{x\in S_i}\kappa_{ij}(x,y)=\pi_j(y).
\end{equation*}

Accordingly, each pair $x\in S_i$ and $y\in S_j$ receives demand $\pi(S_i)\pi(S_j)\kappa_{ij}(x,y)$.  We may split this demand fractionally among paths from $x$ to $y$.  Thus the nonnegative path weights $\Psi_{ij}^{x,y}(\gamma)$ satisfy
\begin{equation}
  \sum_{\gamma\in\operatorname{Paths}(x,y)}
  \Psi_{ij}^{x,y}(\gamma)
  =
  \pi(S_i)\pi(S_j)\kappa_{ij}(x,y).
  \label{eq:partition-flow-demand}
\end{equation}
The length-weighted congestion of the resulting coupled flow $\Psi$ is
\begin{equation*}
  \rho_{\+S}(\Psi)
  :=
  \max_{e\in\+T}
  \frac1{Q(e)}
  \sum_{1\leq i<j\leq k}
  \sum_{x\in S_i}
  \sum_{y\in S_j}
  \sum_{\substack{
    \gamma\in\operatorname{Paths}(x,y)\\
    e\in\gamma
  }}
  \Psi_{ij}^{x,y}(\gamma)\abs{\gamma}.
\end{equation*}

\begin{lemma}[Coupled-flow bound]
  \label{lem:partition-coupled-flow}
  For every coupled flow satisfying \cref{eq:partition-flow-demand},
  \begin{equation*}
    \tau_{\+S}\leq\rho_{\+S}(\Psi).
  \end{equation*}
\end{lemma}

\begin{proof}
  For every $1\leq i<j\leq k$, the marginal conditions on $\kappa_{ij}$ imply
  \begin{equation*}
    \E[\pi_i]{f}-\E[\pi_j]{f}
    =
    \sum_{x\in S_i}\sum_{y\in S_j}
    \kappa_{ij}(x,y)\bigl(f(x)-f(y)\bigr).
  \end{equation*}
  Jensen's inequality therefore bounds the squared difference of the conditional means by
  \begin{equation*}
    \left(\E[\pi_i]{f}-\E[\pi_j]{f}\right)^2
    \leq
    \sum_{x\in S_i}\sum_{y\in S_j}
    \kappa_{ij}(x,y)\bigl(f(x)-f(y)\bigr)^2.
  \end{equation*}
  Multiplying by $\pi(S_i)\pi(S_j)$, summing over $i<j$, and using \cref{eq:restricted-variance-pairwise} and \cref{eq:partition-flow-demand}, we obtain
  \begin{align*}
    \Var[\+S]{f}
    &\leq
    \sum_{1\leq i<j\leq k}
    \sum_{x\in S_i}
    \sum_{y\in S_j}
    \sum_{\gamma\in\operatorname{Paths}(x,y)}
    \Psi_{ij}^{x,y}(\gamma)\bigl(f(x)-f(y)\bigr)^2\\
    &\leq
    \sum_{e=\{a,b\}\in\+T}
    \bigl(f(a)-f(b)\bigr)^2
    \sum_{1\leq i<j\leq k}
    \sum_{x\in S_i}
    \sum_{y\in S_j}
    \sum_{\substack{\gamma\in\operatorname{Paths}(x,y)\\e\in\gamma}}
    \Psi_{ij}^{x,y}(\gamma)\abs{\gamma}\\
    &\leq
    \rho_{\+S}(\Psi)
    \sum_{e=\{a,b\}\in\+T}
    Q(e)\bigl(f(a)-f(b)\bigr)^2\\
    &=
    \rho_{\+S}(\Psi)\+E_P(f,f),
  \end{align*}
  where the second inequality applies Cauchy--Schwarz along each path and the last equality uses \cref{eq:flow-dirichlet}.  Taking the supremum over $f$ proves the result.
\end{proof}

The essential difference from Sinclair's ordinary flow lies in the demand matrix.  The standard construction routes the product demand $\pi(x)\pi(y)$, corresponding to independent endpoints.  Here, the two endpoint marginals are fixed to be $\pi_i$ and $\pi_j$, while their joint distribution may be chosen as any coupling $\kappa_{ij}$.  Once this coupled demand is fixed, the remainder of the proof is the standard path-congestion argument.

Sinclair's multicommodity-flow bound, stated in \cref{lem:sinclair-flow}, is recovered from \cref{lem:partition-coupled-flow} by taking $\+S=\set{\set{x}:x\in\Omega}$ to be the singleton partition.  Each conditional distribution is then a point mass, so the coupling between any two blocks is uniquely determined.  Accordingly, \cref{eq:partition-flow-demand} becomes the product demand $\pi(x)\pi(y)$.

The restricted-terminal construction in \cref{cor:restricted-terminal-flow} has a direct block-level analogue: given a partition $\+S=\set{S_1,\ldots,S_k}$, we may designate one block $S_r$ as the target and route the demand from every other block to $S_r$.
In the application to permanent approximation, the perfect-matching block serves as the target, while each remaining block consists of the near-perfect matchings with a fixed hole pattern.

\begin{corollary}[Target-block coupled flows]
  \label{cor:partition-target-flows}
  Fix $1\leq r\leq k$.  For each $i\neq r$, choose a coupling $\kappa_i$ of $\pi_i$ and $\pi_r$.  For every $x\in S_i$ and $y\in S_r$, let the nonnegative path weights $\Psi_i^{x,y}(\gamma)$ satisfy $\sum_{\gamma\in\operatorname{Paths}(x,y)}\Psi_i^{x,y}(\gamma)=\pi(S_i)\kappa_i(x,y)$.
  Define
  \begin{equation*}
    \rho_{\+S,r}(\Psi)
    :=
    \max_{e\in\+T}
    \frac1{Q(e)}
    \sum_{\substack{1\leq i\leq k\\i\neq r}}
    \sum_{x\in S_i}
    \sum_{y\in S_r}
    \sum_{\substack{\gamma\in\operatorname{Paths}(x,y)\\e\in\gamma}}
    \Psi_i^{x,y}(\gamma)\abs{\gamma}.
  \end{equation*}
  Then
  \begin{equation*}
    \tau_{\+S}\leq\rho_{\+S,r}(\Psi).
  \end{equation*}
\end{corollary}

\begin{proof}
  For every $f:\Omega\to\R$, we have
  \begin{align*}
    \Var[\+S]{f}
    &\leq
    \sum_{\substack{1\leq i\leq k\\i\neq r}}
    \pi(S_i)\left(\E[\pi_i]{f}-\E[\pi_r]{f}\right)^2 \\
    &\leq
    \sum_{\substack{1\leq i\leq k\\i\neq r}}
    \pi(S_i)
    \sum_{x\in S_i}\sum_{y\in S_r}
    \kappa_i(x,y)\bigl(f(x)-f(y)\bigr)^2 \\
    &\leq
    \rho_{\+S,r}(\Psi)\+E_P(f,f).
  \end{align*}
  The first inequality evaluates the variational formula for variance at $\E[\pi_r]{f}$, the second applies Jensen's inequality under each coupling $\kappa_i$, and the third is the same pathwise Cauchy--Schwarz and congestion calculation used in the proof of \cref{lem:partition-coupled-flow}.  Taking the supremum over $f$ completes the proof.
\end{proof}

\subsection{Variance of stationary empirical averages}
\label{sec:proof-restricted-resampling}
In this section, we prove \cref{thm:restricted-resampling}.
  For functions $g, h:\Omega \to \mathbb{R}$, define the inner product $\inner{f}{g}$ as:
  \begin{align*}
      \inner{f}{g} := \E[x\sim \pi]{f(x)g(x)}.
  \end{align*}
 We also use $L^2(\pi)$ to denote the inner product space obtained by $\mathbb{R}^\Omega$ equipped with the inner product $\inner{\cdot}{\cdot}$.
  
  It suffices to prove \cref{item:restricted-resampling-restricted}, since \cref{item:restricted-resampling-standard} follows by taking $\+S$ to be the partition of $\Omega$ into singleton blocks, for which $\mathsf A_{\+S}$ is the identity operator, $\tau_{\+S}=\trel$, and $\+F_{\+S}$ is the space of all functions on $\Omega$.
  Fix $g\in\+F_{\+S}$.  
  Subtracting $\E[\pi]{g}$ changes neither side of the desired inequality, so we may assume that $\E[\pi]{g}=0$.  

  Since $P$ is reversible, it admits an orthonormal eigenbasis $\phi_1=\one,\phi_2,\ldots,\phi_{\abs{\Omega}}$ in $L^2(\pi)$.  
  Note that $\inner{\phi_1}{\phi_i}=\E[\pi]{\phi_i}=0$ for every $i\geq2$.
  Let $\lambda_1=1,\lambda_2,\ldots,\lambda_{\abs{\Omega}}$ be the corresponding eigenvalues.  Nonnegativity of the spectrum and irreducibility imply that $0\leq\lambda_i<1$ for every $i\geq2$.

  Since $g$ is centered, write $g=\sum_{i=2}^{\abs{\Omega}}a_i\phi_i$.  For every $s\geq0$, orthonormality gives
  \begin{equation*}
    \inner{g}{P^s g}
    =\inner{\sum_{i=2}^{\abs{\Omega}}a_i\phi_i}{\sum_{i=2}^{\abs{\Omega}}a_i\lambda_i^s\phi_i}
    =\sum_{i=2}^{\abs{\Omega}}a_i^2\lambda_i^s
    \geq0.
  \end{equation*}

  The geometric-series identity gives
  \begin{equation*}
    \sum_{s=0}^{\infty}P^s g
    =\sum_{i=2}^{\abs{\Omega}}\frac{a_i}{1-\lambda_i}\phi_i.
  \end{equation*}
  Let $u\coloneqq\sum_{s=0}^{\infty}P^s g$.  The preceding representation shows that $u$ is centered, and multiplying it by $I-P$ gives $(I-P)u=g$.  Because the chain is irreducible, this is the unique centered solution of that equation.

  Since $g\in\+F_{\+S}$ is constant within every block, averaging $u$ within the blocks does not change its inner product with $g$.
  Thus $\inner{g}{u}=\inner{g}{\mathsf A_{\+S}u}$.
  Therefore,
  \begin{align*}
    \inner{g}{u}^2
    =\inner{g}{\mathsf A_{\+S}u}^2
    &\leq\E[\pi]{g^2}\E[\pi]{(\mathsf A_{\+S}u)^2}
    \tag{Cauchy--Schwarz inequality}\\
    &=\Var[\pi]{g}\Var[\+S]{u}
    \tag{$\E[\pi]{g}=\E[\pi]{\mathsf A_{\+S}u}=0$}\\
    &\leq\tau_{\+S}\Var[\pi]{g}\,\+E_P(u,u)
    \tag{restricted Poincar\'e inequality}\\
    &=\tau_{\+S}\Var[\pi]{g}\inner{g}{u}.
    \tag{$(I-P)u=g$}
  \end{align*}
  If $g=0$, then the result is immediate.  If $g\neq0$, then $\inner{g}{u}=\+E_P(u,u)>0$.  Cancelling $\inner{g}{u}$ in the preceding inequality yields
  \begin{equation*}
    \sum_{s=0}^{\infty}\inner{g}{P^s g}
    =\inner{g}{u}
    \leq
    \tau_{\+S}\Var[\pi]{g}.
  \end{equation*}

  It remains to translate this bound into a variance estimate for the empirical average.  Stationarity, centering, and the Markov property give, for every $s\geq0$,
  \begin{equation*}
    \Cov\tuple{g(X_0),g(X_s)}
    =\E{g(X_0)\E{g(X_s)\mid X_0}}
    =\E[\pi]{g(X_0)(P^s g)(X_0)}
    =\inner{g}{P^s g}.
  \end{equation*}
  By stationarity, the covariance between two observations depends only on their time lag.  Grouping pairs by this lag and using the nonnegativity established above, we obtain
  \begin{align*}
    \Var{\frac1m\sum_{t=0}^{m-1}g(X_t)}
    &=\frac1{m^2}\sum_{t,r=0}^{m-1}\Cov\tuple{g(X_t),g(X_r)}\\
    &=\frac1{m^2}\tuple{m\inner{g}{g}+2\sum_{s=1}^{m-1}(m-s)\inner{g}{P^s g}}\\
    &\leq\frac1m\tuple{\inner{g}{g}+2\sum_{s=1}^{\infty}\inner{g}{P^s g}}\\
    &\leq\frac2m\sum_{s=0}^{\infty}\inner{g}{P^s g}
    \leq\frac{2\tau_{\+S}}m\Var[\pi]{g}.
  \end{align*}

\section{Faster FPRAS via Hole-Restricted Relaxation Time}
\label{sec:improved-random-weight-refinement}
With the abstract framework in place, 
we now apply \cref{sec:restricted-poincare} to the weight-refinement procedure of
\cite{BSVV08}.  We use the notation introduced in \cref{sec:overview}, with
the convention that $\+P$ and $\+N(u,v)$ are taken in the complete
bipartite graph on $V_1\cup V_2$, all of whose pairs carry positive
activities.  Thus
$\Omega=\+P\cup\bigcup_{(u,v)\in V_1\times V_2}\+N(u,v)$.

Given positive hole weights $w$, let $\pi_{\lambda,w}$ be the distribution on
$\Omega$ assigning unnormalized weight $\lambda(M)$ to $M\in\+P$ and
$\lambda(M)w(u,v)$ to $M\in\+N(u,v)$, and let $Z$ be its normalizing constant.
Recall that the ideal hole weights are
\idealweights* 

\noindent Under $w^*$, the perfect-matching block and each of the $n^2$ near-perfect blocks have equal probability in the stationary distribution $\pi_{\lambda,w^*}$, see \cref{eq:ideal-weights-uniform}.

\subsection{The Metropolis Chain on Perfect and Near-Perfect Matchings}
\label{sec:metropolis-chain}

Fix positive activities $\lambda$ and positive hole weights $w$, and write $\pi:=\pi_{\lambda,w}$.  We use the following lazy Metropolis chain $\MC(\lambda,w)$ on $\Omega$.  Two distinct matchings $M,M'\in\Omega$ are adjacent if one can be obtained from the other by one of the following moves:
\begin{itemize}
  \item \emph{Add/delete:} Delete an edge $(u,v)$ from a perfect matching $M\in\+P$, producing $M\setminus\set{(u,v)}\in\+N(u,v)$, or perform the reverse move by adding the edge joining the two holes of a near-perfect matching.
  \item \emph{Slide fixing the left hole:} If $M\in\+N(u,v)$ and $(x,y)\in M$, replace $(x,y)$ by $(x,v)$.  The resulting matching lies in $\+N(u,y)$.
  \item \emph{Slide fixing the right hole:} If $M\in\+N(u,v)$ and $(x,y)\in M$, replace $(x,y)$ by $(u,y)$.  The resulting matching lies in $\+N(x,v)$.
\end{itemize}
Write $M\sim M'$ when $M$ and $M'$ are adjacent.  The transition matrix $P$ of $\MC(\lambda,w)$ is defined by
\begin{equation}
  P(M,M'):=
  \begin{cases}
    \displaystyle\frac1{4n}\min\set{1,\frac{\pi(M')}{\pi(M)}}, & M\sim M',\\[2mm]
    0, & M\neq M'\text{ and }M\not\sim M',\\[1mm]
    \displaystyle 1-\sum_{M'\neq M}P(M,M'), & M'=M.
  \end{cases}
  \label{eq:metropolis-transition}
\end{equation}
Note that a neighboring proposal and the acceptance ratio
$\pi(M')/\pi(M)$ can both be computed from the activities and hole weights in $O(1)$ time, so each step of the chain can be implemented in $O(1)$ time.
A perfect matching has $n$ neighbors, whereas a near-perfect matching has at most $2n-1$.  Since every off-diagonal transition probability is at most $1/(4n)$, we have $P(M,M)\geq1/2$ for every $M\in\Omega$.  Thus the chain is lazy and, in particular, aperiodic.

For every adjacent pair $M\sim M'$, the Metropolis rule gives
\begin{equation}
  \pi(M)P(M,M')
  =\frac1{4n}\min\set{\pi(M),\pi(M')}
  =\pi(M')P(M',M).
  \label{eq:metropolis-stationary-flow}
\end{equation}
Thus detailed balance holds, so the chain is reversible and $\pi$ is stationary.  Moreover, the underlying transition graph on perfect and near-perfect matchings is connected, and positivity of the activities and hole weights makes every permitted transition have positive probability.  Hence $\MC(\lambda,w)$ is irreducible and $\pi$ is its unique stationary distribution; see also~\cite{BSVV08}.

\subsection{Hole-Restricted Relaxation Time \texorpdfstring{$\tholes$}{tau\_holes}}

The weight-refinement algorithm records only the block containing the current state:
the perfect-matching block $\+P$, or the near-perfect block $\+N(u,v)$ determined by its two holes.  This induces the hole-pattern partition
\begin{equation*}
  \+S_{\mathrm{holes}}
  :=
  \set{\+P}
  \cup
  \set{\+N(u,v):(u,v)\in V_1\times V_2}.
\end{equation*}
For the chain $\MC(\lambda,w)$, let $\tholes=\tau_{\+S_{\mathrm{holes}}}$ denote the restricted relaxation time of the partition $\+S_{\mathrm{holes}}$ from \cref{def:restricted-relax}.
The following lemma bounds $\tholes$; this yields the improved sampling time. 

\begin{lemma}
    \label{lem:restricted-poincare}
    Fix any collection of positive activities $\lambda$, and let $w^*$ be the corresponding ideal hole weights.  The hole-restricted relaxation time of $\MC(\lambda,w^*)$ satisfies
    \begin{equation}
      \tholes\leq16n^3.
      \label{eq:ideal-restricted-poincare}
    \end{equation}
Moreover, if $w$ is any collection of rough hole weights with respect to $\lambda$, then the hole-restricted relaxation time of $\MC(\lambda,w)$ satisfies
    \begin{equation}
      \tholes\leq256n^3.
      \label{eq:rough-restricted-poincare}
    \end{equation}
\end{lemma}

We prove \cref{lem:restricted-poincare} via \cref{cor:partition-target-flows}.  In particular, we set the perfect matchings as the target block $S_r=\mathcal{P}$ and define a coupled flow from $\mathcal{N}(u,v)$ to $\mathcal{P}$.
For this construction, set $w=w^*$ and $\pi=\pi_{\lambda,w^*}$.  For each $(u,v)\in V_1\times V_2$, let
\begin{equation*}
  Z_{uv}:=\lambda(\+N(u,v)),
  \qquad
  Z_\emptyset:=\lambda(\+P).
\end{equation*}
We use the following coupling.
For convenience, for each hole pattern $h$ (either $h=uz$, corresponding to unmatched vertices $u$ and $z$, or $h=\emptyset$, corresponding to a perfect matching), we let $\pi_h$ denote the distribution $\pi$ conditioned on the sampled matching having hole pattern $h$.
For two sets of edges $A$ and $B$, write $A\oplus B$ for their
symmetric difference.
\begin{lemma} \label{couple-I-J}
  Fix $u \in V_1, z \in V_2$.
  Generate the random matching pair $(I,J)$ as follows:
  \begin{enumerate}
  \item sample $I \sim \pi_{uz}$ and $F \sim \pi_\emptyset$ independently;
  \item the symmetric difference $I\oplus F$ is a disjoint union of even alternating cycles and one alternating path $\mathcal{A}$ from $u$ to $z$; put $J := I \oplus \mathcal{A}$.
  \end{enumerate}
  Then $J$ has distribution $\pi_\emptyset$.
\end{lemma}

Fix $u\in V_1$ and $z\in V_2$.  Let $I,F,J$ be generated by the coupling in \Cref{couple-I-J}.
Let $\mathcal{A} = I\oplus J$ be the unique path with endpoints $u,z$ in $I\oplus F$.
In the order from $u$ to $z$, we write the path $\mathcal{A}$ as
\begin{align*}
  u=a_0,b_1,a_1,b_2,\ldots,a_k,b_{k+1}=z,
\end{align*}
where every $a_i\in V_1$ and every $b_i\in V_2$.  Its
$J$-edges are
$(a_{i-1},b_i)$ for $1\leq i\leq k+1$, whereas its $I$-edges are
$(a_i,b_i)$ for $1\leq i\leq k$.
\begin{definition} \label{def:route-I-F}
  We define $\gamma_{I,F}$\footnote{Since $I$ and $F$ uniquely determine $J$, we omit $J$ from the subscript for convenience.} to be a \emph{canonical path} from $I$ to $J$ as follows.  Starting from
  $I$, for $i=1,2,\ldots,k$, replace edge $(a_i,b_i)$ with edge $(a_{i-1},b_i)$.
  After these replacements, the holes are $a_k$ and
  $z=b_{k+1}$.
  We finish the routing by adding the $J$-edge $(a_k,b_{k+1})$.
\end{definition}

Let $P$ denote the transition matrix in \cref{eq:metropolis-transition}.
For each source block $\+N(u,z)$, assign the path $\gamma_{I,F}$ the flow weight
$\pi_{uz}(I)\pi_\emptyset(F)/(n^2+1)$, accumulating weights if the same path arises more than once.
By \cref{couple-I-J}, the two endpoints have marginals $\pi_{uz}$ and $\pi_\emptyset$, so these path weights form a feasible flow in \cref{cor:partition-target-flows} with target block $S_r=\+P$.
Every path has length at most $n$.  Moreover, an undirected transition may be traversed in either orientation.  It follows that
\begin{align} \label{eq:def-congestion}
  \rho_{\+S_{\mathrm{holes}},r}(\Psi)
  \leq \frac{2n}{n^2+1}
  \max_{\alpha\mapsto\beta}
  \frac{1}{\pi(\alpha)P(\alpha,\beta)}
  \sum_{\substack{u\in V_1\\ z\in V_2}}
  \sum_{\substack{I\in\+N(u,z),\ F\in\+P\\
                    \gamma_{I,F}\ni\alpha\mapsto\beta}}
  \pi_{uz}(I)\pi_\emptyset(F),
\end{align}
where the maximum is over feasible directed non-loop transitions of $P$.  The factor $n$ bounds the path length, the factor $(n^2+1)^{-1}$ is the mass of each source block under the ideal weights, and the factor $2$ accounts for the two orientations of an undirected transition.
We next bound this congestion.

\begin{lemma} \label{lem:bound-for-congestion}
  The coupled flow above satisfies $\rho_{\+S_{\mathrm{holes}},r}(\Psi)\leq16n^3$.
\end{lemma}

\cref{lem:bound-for-congestion} is proved in \cref{sec:bound-for-congestion}.
Before proceeding, we identify the technical rationale for the $O(n)$ improvement in the associated congestion bound compared with \cite{BSVV08}.  The key point in \cref{eq:def-congestion} is that every route keeps its right hole $z$ fixed, so a fixed transition can receive flow from only the $n$ source blocks $\mathcal N(u,z)$.  In the cycle-unwinding portion of the paths in~\cite{BSVV08}, this fixed-hole property is lost and $O(n^2)$ source blocks may contribute; eliminating that portion is the factor-$n$ saving in \cref{lem:bound-for-congestion}.

We now prove \cref{lem:restricted-poincare}.
\begin{proof}[Proof of \cref{lem:restricted-poincare}]
  \Cref{eq:ideal-restricted-poincare} follows directly by combining \cref{cor:partition-target-flows} and \cref{lem:bound-for-congestion}.
  It remains to deduce \cref{eq:rough-restricted-poincare} from \cref{eq:ideal-restricted-poincare}.
  Fix positive activities $\lambda$ and rough hole weights $w$ as in the lemma.  
  Let $\pi$ and $P$ denote the stationary distribution and transition matrix of $\MC(\lambda,w)$, 
  and let $\pi^*$ and $P^*$ denote those of $\MC(\lambda,w^*)$.
  Let $Z$ and $Z^*$ be the respective normalizing constants.
  For a perfect matching, the ratio of its unnormalized masses in the two chains is $1$; for a matching in $\+N(u,v)$, this ratio is $w(u,v)/w^*(u,v)$.  Roughness therefore places every such ratio in $[1/2,2]$.
  Summing the unnormalized masses gives $Z^*/2\leq Z\leq2Z^*$.  Consequently,
  \begin{equation*}
    \forall M\in\Omega,\qquad
    \frac{\pi(M)}{\pi^*(M)}
    \in\left[\frac14,4\right].
  \end{equation*}

  The blockwise factor cancels upon conditioning, so the conditional distributions of $\pi$ and $\pi^*$ within every block coincide.  Hence the block-averaging operator $\mathsf A_{\+S_{\mathrm{holes}}}$ is the same for the two chains.
  Using the variational characterization of variance and the bound $\pi\leq4\pi^*$, 
  \begin{align*}
    \forall f:\Omega\to\R,\qquad
    \Var[\pi]{\mathsf A_{\+S_{\mathrm{holes}}} f}
    &=\min_{a\in\R}\E[\pi]{(\mathsf A_{\+S_{\mathrm{holes}}} f-a)^2}\\
    &\leq\E[\pi]{\bigl(\mathsf A_{\+S_{\mathrm{holes}}} f-\E[\pi^*]{\mathsf A_{\+S_{\mathrm{holes}}} f}\bigr)^2}\\
    &\leq4\Var[\pi^*]{\mathsf A_{\+S_{\mathrm{holes}}} f}.
  \end{align*}

  The two chains use the same state space and proposal moves.  
  Since all activities and hole weights are positive, 
  every proposed move has positive acceptance probability under both chains.  
  Hence, for distinct states $M,M'\in\Omega$,
  $P(M,M')>0$ if and only if $P^*(M,M')>0$,
  which means that the two chains have the same transition graph.
  For every pair of distinct adjacent states $M,M'$, 
  combining \cref{eq:metropolis-stationary-flow} and the lower bound $\pi\geq\pi^*/4$ gives
  \begin{equation*}
    \pi(M)P(M,M')
    =\frac1{4n}\min\{\pi(M),\pi(M')\}
    \geq\frac1{16n}\min\{\pi^*(M),\pi^*(M')\}
    =\frac14\pi^*(M)P^*(M,M').
  \end{equation*}
  Since self-loops contribute zero to the Dirichlet form, 
  summing this inequality over adjacent pairs yields $\+E_{P^*}(f,f)\leq4\+E_P(f,f)$.
  Applying the ideal-weight bound in \cref{eq:ideal-restricted-poincare} to $\MC(\lambda,w^*)$, 
  \begin{align*}
    \forall f:\Omega\to\R,\qquad
    \Var[\pi]{\mathsf A_{\+S_{\mathrm{holes}}}f}
    \leq4\Var[\pi^*]{\mathsf A_{\+S_{\mathrm{holes}}} f}
    \leq64n^3\+E_{P^*}(f,f)
    \leq256n^3\+E_P(f,f).
  \end{align*}
  Taking the supremum over all $f$ proves \cref{eq:rough-restricted-poincare}.
\end{proof}

\begin{proof}[Proof of \cref{couple-I-J}]
  We first justify the asserted structure of the symmetric difference $I\oplus F$ in the second bullet.  Regard
  $I\oplus F$ as a graph on $V_1\cup V_2$.  Since $I$ leaves exactly $u$ and
  $z$ unmatched whereas $F$ is perfect, the vertices $u$ and $z$ have degree
  one in $I\oplus F$.  Every other vertex has degree zero if $I$ and $F$
  use the same incident edge there, and degree two otherwise; in the latter
  case, one incident edge belongs to $I$ and the other to $F$.  Hence every
  nontrivial component is a path or a cycle, the edges in each component
  alternate between $I$ and $F$, and the only path component is the path
  $\mathcal A$ whose endpoints are $u$ and $z$.  All remaining components are
  even alternating cycles.

  We now prove that $J$ has distribution $\pi_\emptyset$.  Define the following map:
  \begin{align*}
    \Phi:\+N(u,z)\times\+P
      &\longrightarrow \+N(u,z)\times\+P,\\
    \Phi(I,F)&:=\tp{F\oplus\mathcal A, I\oplus\mathcal A}.
  \end{align*}
  %This map is type preserving.  
  First, we show that $F\oplus\mathcal A \in \+N(u,z)$ and $I\oplus\mathcal A \in \+P$.
  Indeed, along $\mathcal A$, switching removes
  from $F$ the edges incident to the endpoints $u,z$ and replaces the
  $F$-edge at every internal vertex by its incident $I$-edge.  Thus
  $F\oplus\mathcal A\in\+N(u,z)$.  Conversely, the switch inserts into
  $I$ the $F$-edges incident to $u,z$ and leaves every internal vertex matched,
  so $I\oplus\mathcal A\in\+P$.

  Moreover,
  \begin{align*}
    \tp{F\oplus\mathcal A}\oplus
    \tp{I\oplus\mathcal A}=F\oplus I.
  \end{align*}
  Therefore the unique $u$--$z$ path determined by the output pair is again
  $\mathcal A$.  Switching it a second time recovers $(I,F)$, and hence $\Phi$
  is an involution, in particular a bijection.

  The switch also preserves the product activity.  Every edge has the same
  total multiplicity in the two matchings before and after the switch, and
  consequently
  \begin{align*}
    \*\lambda(I)\*\lambda(F)
      =\*\lambda(F\oplus\mathcal A)
       \*\lambda(I\oplus\mathcal A).
  \end{align*}
  Within a fixed block the reweighting factor is constant, so
  \begin{align*}
    \pi_{uz}(I)=\frac{\*\lambda(I)}{Z_{uz}},
    \qquad
    \pi_\emptyset(F)=\frac{\*\lambda(F)}{Z_\emptyset}.
  \end{align*}
  Since $I$ and $F$ are independent, their joint mass is therefore
  $\*\lambda(I)\*\lambda(F)/(Z_{uz}Z_\emptyset)$.  The bijection and the product
  identity show that the push-forward under $\Phi$ has this same product law.
  In particular, for every $J_0\in\mathcal{P}$,
  \begin{align*}
    \Pr{J=J_0}
      &=\sum_{I'\in\+N(u,z)}
        \frac{\*\lambda(I')\*\lambda(J_0)}{Z_{uz}Z_\emptyset}
        =\frac{\*\lambda(J_0)}{Z_\emptyset}
        =\pi_{\emptyset}(J_0). \qedhere
  \end{align*}
\end{proof}

\subsection{Random Weight-Refinement Algorithm}

Using \cref{lem:restricted-poincare} with \cref{thm:restricted-resampling}, we can prove \cref{thm:hole-weight-boosting} with the improved running time for weight refinement.

\begin{proof}[Proof of \Cref{thm:hole-weight-boosting}]
  Fix positive activities $\lambda$ and rough hole weights $w$, and let $\pi:=\pi_{\lambda,w}$.
  Roughness gives
  \begin{equation*}
    \frac12\lambda(\+P)
    \leq
    w(u,v)\lambda(\+N(u,v))
    \leq
    2\lambda(\+P)
  \end{equation*}
  for every $(u,v)$.
  Consequently, $Z\leq(1+2n^2)\lambda(\+P)$, and every block $B\in\+S_{\mathrm{holes}}$ satisfies
  \begin{equation*}
    \pi(B)\geq\frac1{4(n^2+1)}.
  \end{equation*}
  Moreover, \cref{eq:rough-restricted-poincare} gives $\tholes\leq256n^3$.

  Let $M_{\max}$ be a perfect matching of maximum activity.
  Since $\lambda(\+P)\leq n!\lambda(M_{\max})$,
  \begin{equation*}
    \pi(M_{\max})
    \geq
    \frac1{(1+2n^2)n!}.
  \end{equation*}
  Under the roughness assumption, the mixing bound of~\cite[Theorem~4.1]{BSVV08} applies to \cref{eq:metropolis-transition}; its off-diagonal transition probabilities agree with those used there up to an absolute constant factor.
  It therefore produces an initial distribution within total-variation distance $1/8$ of $\pi$, starting from $M_{\max}$, in $O(n^5\log n)$ transitions.
  A maximum-activity perfect matching can be found in lower-order time.

  Let $X_0\sim\pi$, and let $X_1,X_2,\ldots$ be generated by $\MC(\lambda,w)$.
  For every block $B\in\+S_{\mathrm{holes}}$, apply \cref{item:restricted-resampling-restricted} of \cref{thm:restricted-resampling} to its indicator.
  For every $m\geq1$,
  \begin{equation*}
    \Var{\frac1m\sum_{t=0}^{m-1}\one_B(X_t)}
    \leq
    \frac{2\tholes}{m}\pi(B),
  \end{equation*}
  and hence Chebyshev's inequality gives
  \begin{equation*}
    \Pr{\left|\frac1m\sum_{t=0}^{m-1}\one_B(X_t)-\pi(B)\right|>\frac1{10}\pi(B)}
    \leq
    \frac{200\tholes}{m\pi(B)}.
  \end{equation*}
  The preceding bounds on $\tholes$ and $\pi(B)$ allow a common choice $m=O(n^5)$ for which the last probability is at most $1/8$ for every block.
  All block estimates are computed from this same trajectory, since each visited matching increments exactly one block counter.
  If the trajectory instead begins after the preceding burn-in, contraction of total-variation distance shows that each fixed block estimate fails with probability at most $1/4$.

  Let $R$ be the least odd integer satisfying
  \begin{equation*}
    R\geq8\log\frac{n^2+1}{\delta}.
  \end{equation*}
  Run $R$ independent copies of the burn-in and trajectory, and take the median separately for each block.
  Hoeffding's inequality and a union bound show that all $n^2+1$ medians are simultaneously relative-$1/10$ estimates with probability at least $1-\delta$.
  Each copy costs $O(n^5\log n)$ transitions, and each visited matching updates one block counter in constant time.
  Since $R=O(\log(2n/\delta))$, the total running time is
  \begin{equation*}
    T_{\mathcal R}(n,\delta)
    =O\left(
      n^5\log n\,
      \log\left(\frac{2n}{\delta}\right)
    \right).
  \end{equation*}

  Denote the resulting estimates of $\pi(\+P)$ and $\pi(\+N(u,v))$ by $\widehat p_0$ and $\widehat p_{uv}$.
  If any estimate is zero, return the original weights; otherwise, set
  \begin{equation*}
    w'(u,v)
    :=
    w(u,v)\frac{\widehat p_0}{\widehat p_{uv}}.
  \end{equation*}
  Since
  \begin{equation*}
    w^*(u,v)
    =
    w(u,v)\frac{\pi(\+P)}{\pi(\+N(u,v))},
  \end{equation*}
  simultaneous relative-$1/10$ accuracy implies
  \begin{equation*}
    \frac{w'(u,v)}{w^*(u,v)}
    \in
    \left[\frac9{11},\frac{11}{9}\right]
    \subseteq
    \left[\frac1{\sqrt2},\sqrt2\right]
  \end{equation*}
  for every $(u,v)$.
  Thus $w'$ is accurate with probability at least $1-\delta$, as required.
\end{proof}

\section{Congestion Analysis} \label{sec:bound-for-congestion}
In this section, we prove \Cref{lem:bound-for-congestion}.
According to the definition of the Metropolis chain (\Cref{sec:metropolis-chain}) and \cref{eq:def-congestion}, it suffices for us to show the following bound for every transition $\alpha \mapsto \beta$:
\begin{align} \label{eq:congestion-target}
   \sum_{u\in V_1,\ z\in V_2}
  \sum_{\substack{I\in\+N(u,z),\ F\in\+P \\
                   \gamma_{I,F}\ni \alpha\mapsto\beta}}
  \frac{\*\lambda(I)\*\lambda(F)}{Z_{uz}Z_\emptyset}
  \leq 2n(n^2+1)\min\set{\pi(\alpha),\pi(\beta)}.
\end{align}
Indeed, \cref{eq:metropolis-stationary-flow} shows that \cref{eq:congestion-target} bounds the maximum in \cref{eq:def-congestion} by $8n^2(n^2+1)$; multiplying by the crefactor in \cref{eq:def-congestion} gives $16n^3$.

We will use the following result as a black box.
\begin{lemma}[{\cite[Lemma~6.2(2)]{BSVV08}}] \label{switching-inequality}
  For $u,x\in V_1$ and $a,z\in V_2$ such that $u\neq x$ and $a\neq z$,
  \begin{align*}
    \lambda_{(x,a)} Z_{ua}Z_{xz} \leq 2 Z_{uz}Z_\emptyset.
  \end{align*}
\end{lemma}

We prove \cref{eq:congestion-target} by a case analysis on $\abs{\alpha\oplus\beta}$: (1) $\abs{\alpha\oplus\beta} = 2$; (2) $\abs{\alpha\oplus\beta} = 1$.

\paragraph{Case: $\alpha\mapsto\beta$ is a slide ($\abs{\alpha\oplus\beta} = 2$).}
The route in \Cref{def:route-I-F} keeps its right hole fixed and only moves
its left hole.  Hence every slide fixing the left hole has zero load and satisfies
\cref{eq:congestion-target} immediately.  It remains to consider a slide fixing
the right hole that is used by at least one route.  There are unique $v,x\in V_1$ and
$a,z\in V_2$, with $v\neq x$ and $a\neq z$, such that
\begin{align*}
  \alpha\in\mathcal{N}(v,z), & \hspace{.3in} &
  \alpha\setminus\beta=\set{(x,a)},\\
  \beta\in\mathcal{N}(x,z), & & 
  \beta\setminus\alpha=\set{(v,a)}.
\end{align*}
Equivalently,
\begin{align*}
  \beta=\alpha\setminus\set{(x,a)}\cup\set{(v,a)}.
\end{align*}

Fix a possible source hole $u\in V_1$.  Define the encoding
\begin{align} \label{eq:slide-encoding}
  \Psi_u^{\mathrm{sl}}:
  \set{(I,F)\in\+N(u,z)\times\+P:
        \gamma_{I,F}\ni\alpha\mapsto\beta}
  &\longrightarrow\+N(u,a),\\
  (I,F)&\longmapsto
  K:=(I\mathbin\uplus F)
      \setminus\tp{\alpha\mathbin\uplus\set{(v,a)}}. \notag
\end{align}
Here $\uplus$ and $\setminus$ in the formula defining $K$ are multiset union
and multiset subtraction, so edge multiplicities are retained.

We verify that the encoding is well defined.  If
$\alpha\mapsto\beta$ is the $i$th slide on the alternating path, then
\begin{align*}
  v=a_{i-1},\qquad x=a_i,\qquad a=b_i.
\end{align*}
On the path, $\alpha\mathbin\uplus\set{(v,a)}$ contains the $J$-edges
of indices at most $i$ and the $I$-edges of indices at least $i$.  Away from
the path, $\alpha$ agrees with $I$.  Thus the multiset subtraction is valid.
On the path, $K$ contains the $I$-edges of indices less than $i$ and the
$J$-edges of indices greater than $i$; away from the path, $\alpha=I$, so
$K$ agrees with $F$.  Therefore $K$ is a matching with holes $u$ and
$a$, so $K\in\+N(u,a)$.

The encoding is injective for fixed $u$ and fixed
$\alpha\mapsto\beta$.  Indeed, from $\alpha$ and $K$ we recover the multiset
\begin{align*}
  I\mathbin\uplus F
  =\alpha\mathbin\uplus K\mathbin\uplus\set{(v,a)}.
\end{align*}
Every edge of multiplicity two belongs to both $I$ and $F$.  Among the
remaining edges, the unique path with endpoints $u,z$ is oriented from $u$:
its first edge belongs to $F$, since $u$ is unmatched in $I$, and alternation
then recovers both matchings on the path.  Every other nontrivial component is
an even alternating cycle.  The route never changes such a cycle, so
$\alpha$ agrees with $I$ there and determines its two alternating classes.
This reconstructs the ordered pair $(I,F)$ uniquely.

The multiset identity also gives the activity identity
\begin{align} \label{eq:slide-activity}
  \*\lambda(I)\*\lambda(F)
  =\*\lambda(\alpha)\*\lambda(K)\lambda_{(v,a)}.
\end{align}
Consequently, injectivity of the encoding implies
\begin{align} \label{eq:slide-fixed-source}
  &\sum_{\substack{I\in\+N(u,z),\ F\in\+P\\
                    \gamma_{I,F}\ni\alpha\mapsto\beta}}
  \frac{\*\lambda(I)\*\lambda(F)}{Z_{uz}Z_\emptyset}
  \leq
  \frac{\*\lambda(\alpha)\lambda_{(v,a)}Z_{ua}}
       {Z_{uz}Z_\emptyset}.
\end{align}

We compare the last expression with both endpoint probabilities.  Since
$\alpha\in\+N(v,z)$,
\begin{align} \label{eq:slide-alpha-comparison}
  &\frac{1}{\pi(\alpha)}
  \frac{\*\lambda(\alpha)\lambda_{(v,a)}Z_{ua}}
       {Z_{uz}Z_\emptyset}
  =(n^2+1)
    \frac{\lambda_{(v,a)}Z_{ua}Z_{vz}}
         {Z_{uz}Z_\emptyset}.
\end{align}
If $u\neq v$, \Cref{switching-inequality}, with its variable $x$ equal to
$v$, bounds \cref{eq:slide-alpha-comparison} by $2(n^2+1)$.  If $u=v$,
then $i=1$, so this is the first slide.  In this boundary case the explicit map
\begin{align*}
  \+N(v,a)&\longrightarrow\+P,&
  N&\longmapsto N\cup\set{(v,a)}
\end{align*}
is injective, and hence $\lambda_{(v,a)}Z_{va}\leq Z_\emptyset$.
Thus \cref{eq:slide-alpha-comparison} is at most $n^2+1$, and therefore at
most $2(n^2+1)$, in this case as well.

For the other endpoint,
\begin{align*}
  \*\lambda(\beta)
  =\*\lambda(\alpha)\frac{\lambda_{(v,a)}}{\lambda_{(x,a)}}.
\end{align*}
Since $\beta\in\+N(x,z)$, we obtain
\begin{align} \label{eq:slide-beta-comparison}
  &\frac{1}{\pi(\beta)}
  \frac{\*\lambda(\alpha)\lambda_{(v,a)}Z_{ua}}
       {Z_{uz}Z_\emptyset}
  =(n^2+1)
    \frac{\lambda_{(x,a)}Z_{ua}Z_{xz}}
         {Z_{uz}Z_\emptyset}
  \leq 2(n^2+1).
\end{align}
Here \Cref{switching-inequality} applies because a route using this slide has
$u=a_0$ and $x=a_i$ for some $i\geq1$, so $u\neq x$.

It follows from \cref{eq:slide-fixed-source},
\cref{eq:slide-alpha-comparison}, and \cref{eq:slide-beta-comparison} that
the load from each fixed source block is at most
\begin{align*}
  2(n^2+1)\min\set{\pi(\alpha),\pi(\beta)}.
\end{align*}
All source blocks with a right hole different from $z$ have zero load.
Summing over the at most $n$ choices of $u$ proves
\cref{eq:congestion-target} in the slide case.

\paragraph{Case: $\alpha\mapsto\beta$ is an addition ($\abs{\alpha\oplus\beta} = 1$).}
Suppose that
\begin{align*}
  \alpha\in\+N(x,z),\qquad
  \beta=\alpha\cup\set{(x,z)}\in\+P.
\end{align*}
The addition is the last transition of every route that uses it.  The right
hole is fixed throughout the route, so only source blocks $\+N(u,z)$ can
contribute.

For each $u\in V_1$, define the encoding
\begin{align} \label{eq:addition-encoding}
  \Psi_u^{\mathrm{add}}:
  \set{(I,F)\in\+N(u,z)\times\+P:
        \gamma_{I,F}\ni\alpha\mapsto\beta}
  &\longrightarrow\+N(u,z),\\
  (I,F)&\longmapsto
  K:=(I\mathbin\uplus F)\setminus\beta, \notag
\end{align}
where the subtraction is again a multiset subtraction.  Immediately after
the addition the route is at $J$, so $\beta=J$.  On the unique $u$--$z$
path, $\beta$ contains the $J$-edges, and the subtraction leaves the
$I$-edges.  Off that path, $\beta=J=I$, so the subtraction removes the
$I$-edges and $K$ agrees with $F$.  Hence $K$ is a matching with holes $u,z$, proving
that the encoding maps into $\+N(u,z)$.

This encoding is injective.  From $\beta$ and $K$ one recovers
\begin{align*}
  I\mathbin\uplus F=\beta\mathbin\uplus K.
\end{align*}
As in the slide case, doubled edges belong to both inputs, the unique
$u$--$z$ path is oriented from $u$ with its first edge in $F$, and the
untouched cycles are oriented by the edges of $\beta=J=I$ on those cycles.
Thus $(I,F)$ is reconstructed uniquely.  Taking activities gives
\begin{align} \label{eq:addition-activity}
  \*\lambda(I)\*\lambda(F)
  =\*\lambda(\beta)\*\lambda(K).
\end{align}
It follows that, for each $u\in V_1$,
\begin{align} \label{eq:addition-fixed-source}
  &\sum_{\substack{I\in\+N(u,z),\ F\in\+P\\
                    \gamma_{I,F}\ni\alpha\mapsto\beta}}
  \frac{\*\lambda(I)\*\lambda(F)}{Z_{uz}Z_\emptyset}
  \leq \frac{\*\lambda(\beta)}{Z_\emptyset}
  =(n^2+1)\pi(\beta).
\end{align}

To compare the two endpoint probabilities, define the explicit injection
\begin{align*}
  \+N(x,z)&\longrightarrow\+P,&
  N&\longmapsto N\cup\set{(x,z)}.
\end{align*}
It gives $\lambda_{(x,z)}Z_{xz}\leq Z_\emptyset$.  Since
$\*\lambda(\beta)=\lambda_{(x,z)}\*\lambda(\alpha)$, we have
\begin{align*}
  \pi(\beta)
  =\frac{\*\lambda(\beta)}{(n^2+1)Z_\emptyset}
  \leq\frac{\*\lambda(\alpha)}{(n^2+1)Z_{xz}}
  =\pi(\alpha).
\end{align*}
Thus the right-hand side of \cref{eq:addition-fixed-source} equals
$(n^2+1)\min\set{\pi(\alpha),\pi(\beta)}$.  Summing over the at most $n$ possible
choices of $u$ proves \cref{eq:congestion-target} in the addition case.
Deletions have zero load because the final transition of every route is an
addition.  The two cases above therefore cover every transition of the
chain, which proves \cref{eq:congestion-target} and hence
\Cref{lem:bound-for-congestion}.

\bibliographystyle{alpha}
\bibliography{refs}

\end{document}